\documentclass[12pt]{amsart}
\usepackage{hyperref, amsmath,amssymb,amsfonts,amsbsy,mathtools,cite,setspace}
\usepackage{graphicx}
\usepackage{amssymb}
\usepackage{amsthm}
\usepackage{enumerate}
\usepackage{wasysym}
\usepackage{cite}
\usepackage{pgf,tikz}
\usepackage{mathrsfs}
\usetikzlibrary{arrows}
\usepackage{mathtools}

\usepackage{braket}
\usepackage{color}
\usepackage{float}

\def\ket#1{| #1 \rangle}
\def\bra#1{\langle #1 |}
\def\kb#1#2{|#1\rangle\!\langle #2 |}
\def\bk#1#2{\langle #1 |#2\rangle}

\def\be{\begin{eqnarray}}
\def\ee{\end{eqnarray}}
\def\bee{\begin{eqnarray*}}
\def\eee{\end{eqnarray*}}

\newtheorem{defn}{Definition}
\newtheorem{prop}{Proposition}
\newtheorem{thm}{Theorem}
\newtheorem{exa}{Example}
\newtheorem{lem}{Lemma}

\newcommand{\C}{{\mathbb C}}
\newcommand{\Z}{{\mathbb Z}}

\renewcommand{\H}{{\mathcal H}}

\newcommand{\operp}{$\bigcirc$\kern-.91em{$\perp$}}

\newcommand{\spn}{\operatorname{span}}

\def\mvr{{\rm mvr}}

\def\be{\begin{eqnarray}}
\def\ee{\end{eqnarray}}
\def\bee{\begin{eqnarray*}}
\def\eee{\end{eqnarray*}}
\def\ot{\otimes}

\begin{document}

\title[Vector Representations of Graphs and One-Way LOCC]{Vector Representations of Graphs and Distinguishing Quantum Product States with One-Way LOCC}
\author[D.W.Kribs, C.Mintah, M.Nathanson, R.Pereira]{David W. Kribs$^{1,2}$, Comfort Mintah$^{1}$, Michael Nathanson$^3$, Rajesh Pereira$^{1}$}

\address{$^1$Department of Mathematics \& Statistics, University of Guelph, Guelph, ON, Canada N1G 2W1}
%\address{$^2$African Institute for Mathematical Sciences, Muizenberg, Cape Town, 7945, South Africa}
\address{$^2$Institute for Quantum Computing and Department of Physics \& Astronomy, University of Waterloo, Waterloo, ON, Canada N2L 3G1}
%\address{$^4$Perimeter Institute for Theoretical Physics, Waterloo, ON, Canada N2L 2Y5}
%\address{$^5$Canadian Institute for Advanced Research, Toronto, ON, Canada M5G 1Z8}
\address{$^3$Department of Mathematics and Computer Science, Saint Mary's College of California, Moraga, CA, USA 94556}

\begin{abstract}
Distinguishing sets of quantum states shared by two parties using only local operations and classical communication measurements is a fundamental topic in quantum communication and quantum information theory. We introduce a graph-theoretic approach, based on the theory of vector representations of graphs, to the core problem of distinguishing product states with one-way LOCC. We establish a number of results that show how distinguishing such states can be framed in terms of properties of the underlying graphs associated with a set of vector product states. We also present a number of illustrative examples.
\end{abstract}



%\date{ }

\maketitle

\section{Introduction}

In quantum information theory, we frequently attempt to recover classical information that has been encoded into quantum states. If our quantum system consists of multiple physical subsystems, we encounter instances where the information can be recovered with joint measurements on the subsystems but not with local measurements \cite{bennett1999quantum,ghosh2004distinguishability,horodecki2003local,chefles2004condition}. This paradigm makes use of local quantum operations and classical communication (LOCC) and includes many topics such as quantum teleportation and data hiding \cite{Teleportation, terhal2001hiding,eggeling2002hiding}. There is also a growing body of work on the more restricted problem of one-way LOCC, in which parties must perform their measurements in a prescribed order \cite{Walgate-2000,Nathanson-2005,fan2004distinguishability,N13,cosentino2013small,yu2012four,kribs2017operator, kribsquantum2019,lattice2019}.

The corresponding linear algebra problem involves attempting to identify an unknown vector $\ket{\psi}$ from an orthonormal set of vectors $\{ \ket{\psi_k}\}$ in a composite (tensor product) Hilbert space $\H = \H_A \ot \H_B$. In the current discussion, we assume that $\H_A$ and $\H_B$ are finite-dimensional complex inner product spaces; and the initial measurement is a set of rank one operators that sum to the identity operator on $\H_A$. The exploration of  connections with linear algebra goes back to the origins of quantum information theory, and all of the citations above can be seen in this light. Our recent work to relate one-way LOCC to operator systems and algebras \cite{kribs2017operator,kribsquantum2019,lattice2019} is a continuation of this exploration.

One of the most unexpected phenomena in quantum information is that of ``nonlocality without entanglement,'' originally identified by Bennett et al. \cite{bennett1999quantum}. It says that a set of states $\{ \ket{\psi_k}\}$ in $\H_A \ot \H_B$ can fail to be locally distinguishable even if each of them is a product state, i.e.  for each $k$, we have $\ket{\psi_k} =  \ket{\psi_k^A}\ot \ket{\psi_k^B} \in \H_A \ot \H_B$. This implies that quantum entanglement is not the only way to take advantage of the peculiarities of quantum information.  For any set of product states, we can ask whether it exhibits this phenomenon. The local relationships between product states are modeled using the confusability graph (as identified in \cite{duan2013zero}), which  arises in the study of vector representations of graphs  \cite{lovasz1989orthogonal,lovasz2000correction,fallat2007minimum,booth2008minimum,fallat2011variants,barioli2010zero,booth2011minimum,barioli2011minimum }. %The theory has found a role in likely every area of contemporary mathematical sciences, with underlying graphs identified in the background of more complicated mathematical structures, problems, or relationships. In such applications, solutions to problems can often be framed in terms of more easily computable and identifiable graph-theoretic properties and conditions.

In this paper, we apply graph-theoretic techniques, including those associated with vector representations of graphs, to the study of LOCC quantum state distinguishability. We specifically focus on the core problem of one-way LOCC distinguishability of product states, identifying and clarifying new structure for the distinguishability of such sets of states based on associated vector graph representations.

The paper is organized as follows. In the next section we present requisite preliminaries from graph theory, along with the vector representations of graphs that arise from sets of quantum product states and an illustrative one-way LOCC example. We then initiate our analysis in Section~3, establishing that one-way distinguishability is equivalent to the existence of a graph clique cover with nice properties. Building on this result, in Section~4 we show that under certain conditions if a set of product states can be distinguished by one way LOCC with either party going first then it can be distinguished with product measurements.  We also provide an example to show that this is false in general. In Section~5 we consider one-way LOCC implications when the underlying graphs have extra properties as identified in graph theory; in particular, chordal or tree structures.
The final section includes an extended analysis for the important special case of domino states \cite{bennett1999quantum,cohen2017general,zhang2014nonlocality,zuo2018new}.  Examples are included throughout our presentation, as are discussions on the differences encountered when different parties go first in a one-way LOCC protocol.

\section{Alice and Bob, and Vector Representations of Graphs}

We begin by recalling basic notation and nomenclature from graph theory, drawing on various entrance points into the literature on the subject \cite{lovasz1989orthogonal,lovasz2000correction,fallat2007minimum,booth2008minimum,fallat2011variants,barioli2010zero,booth2011minimum,barioli2011minimum }.

Let $G = (V,E)$ be a {\it simple graph} with vertex set $V$ and edge set $E$. For $v,w\in V$, we write $v \sim w$ if the edge $\{ v,w \}\in E$.  The {\it complement} of $G$ is the graph $\overline{G} = (V, \overline{E})$, where the edge set $\overline{E}$ consists of all two-element sets from $V$ that are not in $E$. Another graph $G'$ is a {\it subgraph} of $G$, written $G' \leq G$, if $V' \subseteq V$ and $E' \subseteq E$ with $v,w\in V'$ whenever $\{ v,w\} \in E'$. For a subset of the vertices $V' \subset V$, $G' = (V', E')$ is the {\it induced subgraph} of $G$ on $V'$ when $E' = \{ \{v,w\} \in E: v,w \in V'\}$.

Some fundamental graphs for fixed $n\geq 1$ include: the {\it path graph} $P_n = (\{v_1,\dots , v_n    \}, E)$ such that $E = \{ \{v_i, v_{i+1}\} : 1 \leq i \leq n-1 \}$; the {\it cycle graph}  $C_n = (\{v_1,\dots , v_n    \}, E)$ such that $E = \{ \{v_i, v_{i+1}\} : 1 \leq i \leq n-1 \} \cup \{ v_n, v_1\}$; and the {\it complete graph} on $n$-vertices, $K_n = (\{v_1,\dots , v_n    \}, E)$ such that $E = \{ \{v,w\} : v\neq w \in V \}$.

%\begin{defn}
%Given a graph $G= (V,E)$, a function $\phi: V \rightarrow \mathbb{C}^d$ is an {\it orthogonal representation} of $G$ if for all vertices $v_i \ne v_j\in V$,
%\be
%v_i \not\sim v_j \iff \langle \phi(v_i), \phi(v_j) \rangle = 0 .
%\ee

%The {\it minimum vector rank} of $G$, denoted $\mvr(G)$, is the smallest $d$ such that $G$ has an  orthogonal representation in $\mathbb{C}^d$.
%\end{defn}

%Orthogonal representations were discussed in, e.g., \cite{fallat2007minimum,lovasz1989orthogonal}.  If every vertex of $G$ is part of at least one edge, then the minimum vector rank coincides with the {\it minimum semidefinite rank} of the graph \cite{booth2008minimum,booth2011minimum}, and we will be applying the literature of minimum semidefinite rank to prove our results. Note that these quantities are defined with respect to the underlying field $\C$, and there are examples where the minimum rank increases if we restrict ourselves to real vectors. Note also the biconditional built into the definition, which is stronger than conditions for graph colouring. This allows us to uniquely define the graph associated with a function $\phi$. %There are times when discussing distinguishability questions in which we wish to weaken the condition to simply $v_i \not\sim v_j \Longrightarrow  \langle \phi(v_i), \phi(v_j) \rangle = 0$; for clarity, we refer to these as {\it weak orthogonal representations}.

\begin{defn}
Given a graph $G= (V,E)$, a function $\phi: V \rightarrow \mathbb{C}^d\backslash \{0\} $ is an {\it orthogonal representation} of $G$ if for all vertices $v_i \ne v_j\in V$,
\be
v_i \not\sim v_j \iff \langle \phi(v_i), \phi(v_j) \rangle = 0 .
\ee

The {\it minimum vector rank} of $G$, denoted $\mvr(G)$, is the smallest $d$ such that $G$ has an  orthogonal representation in $\mathbb{C}^d$.
\end{defn}

Orthogonal representations were discussed in, e.g., \cite{fallat2007minimum,lovasz1989orthogonal}. The requirement that $\phi(v) \ne 0$ distinguishes the minimum vector rank from the  {\it minimum semidefinite rank} of the graph \cite{booth2008minimum,booth2011minimum}.  If every vertex of $G$ is part of at least one edge, then these two quantities coincide, and we will be applying the existing literature of minimum semidefinite rank to prove our results. Note that these quantities are defined with respect to the underlying field $\C$, and there are examples where the minimum rank increases if we restrict ourselves to real vectors. Note also  the biconditional built into the definition, which is stronger than conditions for graph colouring. This allows us to uniquely define the graph associated with a function $\phi$.

We can now introduce a graph perspective to the setting of LOCC through orthogonal representations. In the LOCC state distinguishability context, we typically have two parties, called Alice and Bob, each with a Hilbert space $\mathcal H_A$, $\mathcal H_B$, and a set of states in $\mathcal H_A \otimes \mathcal H_B$ that they wish to distinguish only using local (quantum) measurement operations on their respective systems and classical communication between them. Here we will focus on the important case of one-way LOCC, where Alice and Bob measure in a prescribed order, with one party communicating the results of their measurement to the other, allowing for the second party to complete the measurement based on that result. We also consider the situation in which the states are all product states; that is, states of the form $\ket{\psi^A_k} \otimes \ket{\psi^B_k}$. Measurements will thus be given for Alice by the (rank-one) projections onto the states $\{ \ket{\psi^A_k} \}_k$, with measurement outcomes corresponding to these states, and similarly for Bob and the states $\{ \ket{\psi^B_k} \}_k$. (Such measurements are rank-one versions of von Neumann positive operator valued measures, or POVM's.)

\begin{defn}
Suppose we are given a set of product states $\{ \ket{\psi^A_k}\otimes \ket{\psi^B_k} \}_{k=1}^r$ on $\mathcal H_A \otimes \mathcal H_B$. The graph of these states from Alice's perspective is the unique graph $G_A$ with vertex set $V = \{ 1,2, \ldots, r\}$ such that the map $k \mapsto  \ket{\psi^A_k}$ is an orthogonal representation of $G_A$. Likewise, the graph of the states from Bob's perspective is the graph $G_B$ with vertex set $V$ such that $k \mapsto \ket{\psi^B_k}$ is an orthogonal representation of $G_B$.
\end{defn}

%Note that by definition the obvious maps on $V_A$ and $V_B$ generated by this identification yield orthogonal representations of $G_A$ and $G_B$.

Consider the following illustrative example, which we also analyze from the one-way LOCC perspective as a prelude to what follows. We will always use the standard notation $\{\ket{0}, \ket{1}, \ldots , \ket{d-1}  \}$ for a fixed orthonormal basis of $\C^d$.

\begin{exa}
{\rm
Consider the following set of five (unnormalized) states in $\mathcal H_A \otimes \mathcal H_B = \C^4 \ot \C^3:$
\begin{align*}
\ket{\psi_1} &= \ket{0} \otimes   \left(   \ket{0} +  \ket{2} \right)\\
\ket{\psi_2} &= \left(   \ket{0} + \ket{1} \right) \otimes \ket{1}  \\
\ket{\psi_3} &= \left(   \ket{1} + \ket{2} \right) \otimes    \ket{2} \\
\ket{\psi_4} &= \left(   \ket{2} + \ket{3} \right) \otimes  \left(   \ket{0} - \ket{1} \right)  \\
\ket{\psi_5} &= \ket{3}  \otimes   \left(   \ket{0} + \ket{1}+ \ket{2} \right) .
\end{align*}

Looking at the set of Alice's vectors, we can see it as an orthogonal representation of the graph $G_A = P_5$ in $\C^4$; with $\ket{\psi_1^A} = \ket{0}$,  $\ket{\psi_2^A} = \ket{0}+ \ket{1}$, $\ket{\psi_3^A} = \ket{1} + \ket{2}$, $\ket{\psi_4^A} = \ket{2} + \ket{3}$, $\ket{\psi_5^A} = \ket{3}$. The set of Bob's vectors is an orthogonal representation of the `house' graph $G_B = \overline{P}_5$ in $\C^3$ depicted in Figure~1.

\begin{figure}[H]
\definecolor{ududff}{rgb}{0.30196078431372547,0.30196078431372547,1.}
\begin{tikzpicture}[line cap=round,line join=round,>=triangle 45,x=1.0cm,y=1.0cm]
\clip(-2.5,-1.32) rectangle (3.64,4.5);
\draw [line width=2.pt] (-0.16,1.)-- (0.84,1.);
\draw [line width=2.pt] (-0.16,1.)-- (-0.16,2.);
\draw [line width=2.pt] (-0.16,2.)-- (0.84,2.);
\draw [line width=2.pt] (0.84,2.)-- (0.84,1.);
\draw [line width=2.pt] (-0.16,2.)-- (0.36,2.84);
\draw [line width=2.pt] (0.36,2.84)-- (0.84,2.);
\draw (-1.16,2.50) node[anchor=north west] {$| \psi_{1}^B\rangle$};
\draw (0.7,1.04) node[anchor=north west] {$| \psi_{2}^B\rangle$};
\draw (-0.02,3.5) node[anchor=north west] {$| \psi_{3}^B\rangle$};
\draw (-0.96,1.06) node[anchor=north west] {$| \psi_{4}^B\rangle$};
\draw (0.90,2.5) node[anchor=north west] {$| \psi_{5}^B\rangle$};
\begin{scriptsize}
\draw [fill=ududff] (-0.16,1.) circle (2.5pt);
\draw [fill=ududff] (0.84,1.) circle (2.5pt);
\draw [fill=ududff] (-0.16,2.) circle (2.5pt);
\draw [fill=ududff] (0.84,2.) circle (2.5pt);
\draw [fill=ududff] (0.36,2.84) circle (2.5pt);
\end{scriptsize}
\end{tikzpicture}
\caption{Complement of $P_{5}$, represented in $\mathbb{C}^{3}$ in Example~1.}
\end{figure}
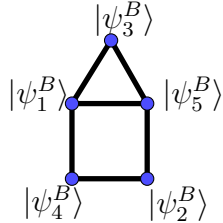

Notice that these states can be distinguished with one-way LOCC with Alice going first: If Alice measures in the standard basis given by $\{  \ket{0}, \ket{1}, \ket{2}, \ket{3}  \}$, then for each of her possible outcomes from this set, there are only two remaining possibilities; and these are orthogonal from Bob's perspective. In general, in order to distinguish our states with one-way LOCC, every measurement outcome of Alice's needs to eliminate possibilities so that the remaining possible states are mutually orthogonal on Bob's side. For instance, in this example if Alice gets a measurement outcome of $\ket{0}$, Bob is left to distinguish between the vectors $\lbrace \ket{0}+\ket{2}, \ket{1}  \rbrace$, which are mutually  orthogonal.

Note that this analysis does not tell us whether we can distinguish with Bob measuring first. We will discover in the next section that this is not possible.

%Let $B$ be the graph of the states from Bob's perspective, and define $G = \overline{B}$. (In many cases, like the current one, $G$ is the graph on Alice's side, but this is not strictly guaranteed.) Let $\phi: V \rightarrow \C^d$  be the association of vertices with  Alice's states. This is a weak representation of $G$.
}
\end{exa}

\section{Graph Clique Covers and One-Way LOCC}

We will now attempt to categorize the one-way LOCC distinguishability of states in terms of their corresponding graphs. We begin with the notion of a clique cover.
\begin{defn}
Given a graph $G = (V,E)$. A set of graphs $\{ G_i = (V_i, E_i) \}$ covers $G$ if $V = \cup_i V_i$ and $E = \cup_i E_i$.

A collection of graphs $\{G_i \}$ is a clique cover for $G$ if $\{G_i \}$ covers $G$ and if each of the $G_i$ is a complete graph (clique).

The clique cover number $\mathrm{cc}(G)$ is the smallest possible number of subgraphs contained in a clique cover of $G$.
\end{defn}

A clique cover can be thought of as a collection of (not necessarily disjoint) induced subgraphs of $G$, each of which is a complete graph. It is a cover if every edge is contained in at least one of the cliques.

In the example of the previous section, the only clique cover of $P_5$ is the set of edges; while the complement $\overline{P_5}$ can be clique covered with a three-cycle and three individual edges, as seen in Figure \ref{Figure2}. Hence, $\mathrm{cc}(P_5) = \mathrm{cc}(\overline{P_5}) = 4$.

\begin{figure}[H]
\begin{center}
\definecolor{ududff}{rgb}{0.30196078431372547,0.30196078431372547,1.}
\begin{tikzpicture}[line cap=round,line join=round,>=triangle 45,x=1.0cm,y=1.0cm]
\clip(-2.24,-2.54) rectangle (6.58,1.94);
\draw [line width=2.pt] (-0.84,1.04)-- (-0.84,0.04);
\draw [line width=2.pt] (0.34,1.)-- (0.34,0.);
\draw [line width=2.pt] (1.42,0.98)-- (1.42,0.06);
\draw [line width=2.pt] (3.,0.)-- (4.,0.);
\draw [line width=2.pt] (3.,0.)-- (3.46,0.96);
\draw [line width=2.pt] (3.46,0.96)-- (4.,0.);
\draw (2.6,0.02) node[anchor=north west] {$| \psi_{1}^B\rangle$};
\draw (-0.06,1.92) node[anchor=north west] {$| \psi_{2}^B\rangle$};
\draw (3.78,0.02) node[anchor=north west] {$| \psi_{3}^B\rangle$};
\draw (-1.22,0.08) node[anchor=north west] {$| \psi_{4}^B\rangle$};
\draw (3.1,1.78) node[anchor=north west] {$| \psi_{5}^B\rangle$};
\draw (-0.02,0.04) node[anchor=north west] {$| \psi_{4}^B\rangle$};
\draw (1.04,1.92) node[anchor=north west] {$| \psi_{5}^B\rangle$};
\draw (-1.24,1.94) node[anchor=north west] {$| \psi_{1}^B\rangle$};
\draw (1.06,0.08) node[anchor=north west] {$| \psi_{2}^B\rangle$};
\begin{scriptsize}
\draw [fill=ududff] (-0.84,1.04) circle (2.5pt);
\draw [fill=ududff] (-0.84,0.04) circle (2.5pt);
\draw [fill=ududff] (0.34,1.) circle (2.5pt);
\draw [fill=ududff] (0.34,0.) circle (2.5pt);
\draw [fill=ududff] (1.42,0.98) circle (2.5pt);
\draw [fill=ududff] (1.42,0.06) circle (2.5pt);
\draw [fill=ududff] (3.,0.) circle (2.5pt);
\draw [fill=ududff] (4.,0.) circle (2.5pt);
\draw [fill=ududff] (3.46,0.96) circle (2.5pt);
\end{scriptsize}
\end{tikzpicture}
\end{center}
\caption{Clique cover of $\overline{P_{5}}$ with labels as represented in Example~1.}
\label{Figure2}
\end{figure}
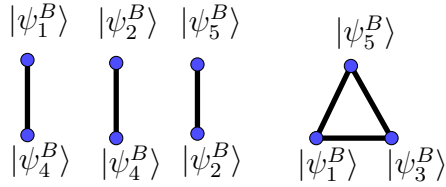

Our first result shows that a set of product states can be perfectly distinguished with one-way LOCC precisely when there is a clique cover with nice properties.

\begin{thm} \label{thm: one-way LOCC graphs} Given a set of product states in $\H_A \ot \H_B$, let $G_A$ and $G_B$ be the graphs of the states from Alice and Bob's perspectives, respectively.  Let $\phi: V_A \rightarrow \H_A$  be the association of vertices with  Alice's states and assume that the set $\{\phi(v): v\in V\}$ spans $\H_A$.

Then the states are distinguishable with one-way LOCC with Alice measuring first if and only if there exists
\begin{itemize}
\item[(1)] a graph $G$ satisfying $G_A \le G \le \overline{G_B}$,
\item[(2)] a clique cover $\{V_j \}_{j = 1}^k$ of $G$, and,
\item[(3)] a POVM $\{Q_j\}$ on $\mathcal H_A$ such that for all $v\in V_A$, $Q_j\phi(v) \ne 0$ implies that $v \in V_j$.
%a direct sum decomposition $\H_A = \oplus_{j = 1}^k {\mathcal S}_j$ with the property that for all $v\in V_A$, the support of $\phi(v)$ is contained in $\oplus_{\{j: v \in V_j\}} {\mathcal S}_j$. 
\end{itemize}
%In particular, the number of cliques $k$ must be less than or equal to the dimension of $\H_A$.
\end{thm}

\begin{proof} One direction of the proof is straightforward: suppose we have a graph $G$ with $G_A \leq G \leq G_B$, a clique cover $\{ V_j \}_{j=1}^k$ of $G$, and 
%the decomposition $\H_A = \oplus_{j = 1}^k {\mathcal S}_j$ along with 
$\{Q_j\}$ with the support assumption of (3). 
%For each $j$ define $Q_j$ as the projection of $\mathcal H_A$ onto ${\mathcal S}_j$, so $\{Q_j\}$ is a (von Neumann) POVM on $\H_A$. 
If Alice gets the outcome $j$ from the associated measurement, then $Q_j\phi(v) \ne 0$, which implies that $v \in V_j$. Since the vertices in $V_j$ form a clique in $G \le \overline{ G_B}$, they form a disconnected set in $G_B$, reflecting the fact that they are mutually orthogonal. Hence, Bob can distinguish them once he knows Alice's outcome.

For the other direction of the proof, let $\{Q_j\}_{j=1}^k$ be a measurement on Alice's system that allows Bob to complete a perfect discrimination of their states. Then for each $Q_j$, we define $V_j = \{ v \in V: Q_j \phi(v) \ne 0 \}$. It is necessary that the vertices in $V_j$ form a clique in $\overline{G}_B$ for Bob to be able to distinguish the remaining possibilities. Define  $G$ to be the union (both vertices and edges) of the cliques induced by the $V_j$. By construction, this is a subgraph of $\overline{G}_B$ and the $V_j$ form a clique cover of $G$.

On the other hand, if $\bk{\phi(u)}{\phi(v)} \ne 0$, then $\bra{\phi(u)}Q_j\ket{\phi(v)} \ne 0$ for some $j$, which means that $u$ and $v$ are both in $V_j$.  Hence every edge in Alice's graph $G_A$ is contained in one of the cliques determined by some $V_j$. Thus we see that $G_A$ is a subgraph of $G$, and we get $G_A \le G \le \overline{G}_B$ as desired.
\end{proof}
%For each measurement operator $Q_j$, define ${\mathcal R}_j$ to be the range of $Q_j$. Then we can define $ {\mathcal S}_1 = {\mathcal R}_1$ and for all $j > 1$,
%\bee  {\mathcal S}_j ={\mathcal R}_j  \cap \left( \bigcap_{i = 1}^{j-1} {\mathcal R}_i^\perp\right).  \eee
%By definition, this means that if $i<j$ then  ${\mathcal S}_j$ is a subspace of  ${\mathcal R}_i^\perp$ and ${\mathcal S}_i$ is a subspace of ${\mathcal R}_i$, and so the subspaces are mutually orthogonal. Also, by construction, each ${\mathcal R}_j \subseteq \oplus_{i = 1}^j {\mathcal S}_i$. Hence  ${\mathcal R}_j \subseteq  \oplus_{i = 1}^k {\mathcal S}_i$ for every $j$. Since $\sum_j Q_j = I$, the linear span of the $\{{\mathcal R}_j\}$ is all of $\H_A$, which implies that
%\bee
%\H_A = \oplus_{j = 1}^k {\mathcal S}_j .
%\eee
%Finally, we see that for each vertex $v$, if $\phi(v)$ has a nontrivial component in ${\mathcal S}_j$ then it has a component in ${\mathcal R}_j$, implying that $v \in V_j$. Thus, the support of $\phi(v)$ is contained in $\oplus_{j: v \in V_j} {\mathcal S}_j$.

\begin{exa}
{\rm
Returning to our example from the previous section in light of the theorem, the graph $G$ is the path $P_5$, and $G_A = P_5 = \overline{G_B}$.  The clique cover is simply the collection of edges, with corresponding subspaces of $\mathcal H_A$ as in the theorem given by:
\bee
V_0 = \{ v_1, v_2 \} &\qquad& {\mathcal S}_0 =  \{  \left(   \ket{1} + \ket{2} \right), \left(   \ket{2} + \ket{3} \right),\ket{3} \}^\perp = \spn \{ \ket{0} \} \\
V_1 = \{ v_2, v_3 \} &\qquad& {\mathcal S}_1 =  \{    \ket{0}, \left(   \ket{2} + \ket{3} \right),\ket{3} \}^\perp = \spn \{   \ket{1} \} \\
V_2= \{ v_3, v_4 \} &\qquad& {\mathcal S}_2 =  \{    \ket{0},  \left(   \ket{0} + \ket{1} \right),  \ket{3} \}^\perp = \spn \{   \ket{2} \} \\
V_3 = \{ v_4, v_5 \} &\qquad& {\mathcal S}_3 =  \{    \ket{0}, \left(   \ket{0} + \ket{1} \right),   \left(   \ket{1} + \ket{2} \right) \}^\perp = \spn \{   \ket{3} \}. \\
\eee
This tells us that we should measure in the standard basis, matching our earlier strategy.  Note that if instead we have Bob measure first, we can see right away that this is not possible. The graph corresponding to Bob's states is the house graph; and a minimum clique cover contains four cliques. Since we are in $\C^3$, there is no clique cover of size less than or equal to $d$ and it is not possible to meet the conditions of the theorem.
In conclusion, if we are to distinguish the states in Example~1 using one-way LOCC, it must be with Alice measuring first.
}
\end{exa}

\section{Product Measurements and One-Way LOCC}

We can build on Theorem~\ref{thm: one-way LOCC graphs} to derive the following important and intuitive consequence, in the case that order of the parties does not matter in a one-way LOCC protocol with an extra condition on the states.

%\marginpar{\tiny DK: Probably there is a nice operator algebra cum LOCC interpretation of the direct sum tensor product decomposiiton in this proof (there is a natural algebra and its commutant associated with the decomp). I can think about it some more.}

\begin{thm}\label{thm: Two 1-LOCC imply 0-LOCC}
Suppose we have a set of product states  $\{ \ket{\psi_i^A} \ot  \ket{\psi_i^B}\}$ in $\H_A\ot \H_B$, and let $G_A$ and $G_B$ be the graphs for Alice and Bob, respectively. Suppose further that $G_A = \overline{G_B}$; that is, for every $i \ne j$,
\[
\bk{\psi_i^A}{\psi_j^A} = 0 \iff  \bk{\psi_i^B}{\psi_j^B} \ne 0 .
\]

Then it is possible to distinguish the states with one-way LOCC with Alice going first as well as with one-way LOCC with Bob going first if and only if it is possible to distinguish the states with a product measurement.
\end{thm}

\begin{proof}
The states can automatically be distinguished with one-way LOCC in either direction if the states can be distinguished with a product measurement.

For the other direction, suppose it is possible to distinguish the states with one-way LOCC with Alice going first as well as with one-way LOCC with Bob going first.
Then from two applications of Theorem~1, we have that $\H_A = \oplus_j {\mathcal S}_j$ and similarly, $\H_B = \oplus_i{\mathcal O}_i$. This implies that the entire Hilbert space has a product decomposition $\H = \H_A \otimes \H_B = \oplus_{i,j} ({\mathcal S}_j \otimes {\mathcal O}_i)$, with the corresponding graph implications of the result.

Suppose that Alice and Bob each perform their one-way measurements and receive the outcomes ${\mathcal S}_j$ and ${\mathcal O}_i$. This means that  $v \in V_j \cap W_i$, where $V_j$ induces a clique in $\overline{G_B}$ and $W_i$ induces a clique in $\overline{G_A}$.
%Our assumption that  $\bk{\phi_i^A}{\phi_j^A} = 0 \iff \bk{\phi_i^B}{\phi_j^B} \ne 0$ is equivalent to assuming that $G_A  =\overline{G}_B$.
Suppose that two vertices are each contained in $V_j \cap W_i$. Then they are connected by an edge in both $\overline{G}_B$ and $\overline{G}_A = G_B$, which is a contradiction. Hence for each pair $i,j$, we have $\vert V_j \cap W_i \vert \le 1$, which means that Alice and Bob can determine the identity of their state without further communication.
\end{proof}

The example below shows that the assumption $\overline{G}_A = G_B$ is necessary in the statement of Theorem \ref{thm: Two 1-LOCC imply 0-LOCC}. If $\overline{G}_A \ne G_B$ , then Alice and Bob have extra flexibility, which weakens the assumption of one-way LOCC in both directions.

\begin{exa}\label{Ex: Redundant Orthogonality Qutrits}
{\rm
Consider the following states in $\C^3 \ot \C^3$:
\begin{align*}
\ket{\psi_1} &= \ket{0} \otimes   \ket{0}\\
\ket{\psi_2} &= \left(   \ket{1} + \ket{2} \right)  \otimes   \ket{0}\\
\ket{\psi_3} &=\left(   \ket{1} -\ket{2} \right) \otimes   \ket{0}\\
\ket{\psi_4} &=\ket{0} \ot  \left(   \ket{1} + \ket{2} \right) \\
\ket{\psi_5} &=\ket{0} \ot \left(   \ket{1} -\ket{2} \right) \\
\ket{\psi_6} &= \ket{1} \otimes   \ket{1}\\
\ket{\psi_7} &= \ket{2} \otimes   \ket{2}.
\end{align*}

Observe in this case that $G_A = G_B = C_3 \cup C_4$ is the disjoint union of two cycle graphs, with respective labelling determined by the ordered sets:
\[
\{ \ket{\psi^A_1},  \ket{\psi^A_4},  \ket{\psi^A_5} \}, \quad \{ \ket{\psi^A_2},  \ket{\psi^A_6},  \ket{\psi^A_3},  \ket{\psi^A_7} \},
\]
and
\[
\{ \ket{\psi^B_1},  \ket{\psi^B_2},  \ket{\psi^B_3} \}, \quad \{ \ket{\psi^B_4},  \ket{\psi^B_6},  \ket{\psi^B_5},  \ket{\psi^B_7} \}.
\]
On the other hand, $\overline{G_A}$ is a graph not isomorphic to this graph, given by the disjoint union of three isolated vertices $(\{ 1\}, \{ 4 \}, \{ 5 \})$ and two isolated edges $( \{ 2,3\}, \{ 6,7\} )$.

These states are distinguishable with one-way LOCC: Alice measures in the basis $\{  \ket{0} ,\left(   \ket{1} \pm \ket{2} \right) \}$. If she gets the outcome $\ket{0}$, then Bob completes the measurement using the same measurement that Alice did; but if she gets either of the remaining outcomes, Bob should measure in the standard basis. It is clear that these states are symmetric with respect to Alice and Bob, so they can also be distinguished with one-way LOCC with Bob going first and corresponding measurement adjusments.
}
\end{exa}

\begin{prop}
The 7 states described in Example \ref{Ex: Redundant Orthogonality Qutrits} cannot be distinguished with a product measurement, even though they can be distinguished with one-way LOCC in both directions.
\end{prop}

\begin{proof}
This fact follows from the observation that Alice's initial measurement for one-way LOCC is unique. Her measurement operators must eliminate at least four possibilities so that Bob has only three orthogonal possibilities remaining. This means that the span of these four vectors cannot be all of $\C^3$, since they are in the kernel of a nonzero operator. There are only 5 sets of 4 of the $\ket{\psi_i}$ that do not span all of $\C^3$. It must also be true that the complementary sets of 3 vectors must be mutually orthogonal on Bob's side. This eliminates 2 of the sets, leaving us with the following three sets of vectors:
\bee
\{ \ket{\psi_4} ,\ket{\psi_5} ,\ket{\psi_6} ,\ket{\psi_7} \}, \,\,\, \{ \ket{\psi_1} ,\ket{\psi_2} ,\ket{\psi_3} ,\ket{\psi_5} \}, \,\,\, \{ \ket{\psi_1} ,\ket{\psi_2} ,\ket{\psi_3} ,\ket{\psi_4} \}.
\eee
These correspond to the measurement $\{  \ket{0} ,\left(   \ket{1} + \ket{2} \right),\left(   \ket{1} - \ket{2} \right) \}$ described above. There are no other options for Alice's initial measurement or for Bob's. As we noted, the second measurement is dependent on the outcome of the first; there is no way to accomplish this without knowing the other outcome.
\end{proof}

In table form, we can see the possibilities with Example~3. The rows and columns represent Alice and Bob's measurements, respectively. The table entries correspond to the possible states obtained amongst $\{ \ket{\psi_1}, \ldots , \ket{\psi_7}\}$:
\bee
\begin{array}{c|c|c|c|c}
& \ket{0}& \ket{1} + \ket{2} & \ket{1} - \ket{2} & \cr \hline
\ket{0} &1&4 & 5  & \{ 1,4,5\} \cr\hline
 \ket{1} + \ket{2}  &  2& 6,7 & 6,7 & \{ 2,6,7\}   \cr\hline
 \ket{1} - \ket{2}  &  3& 6,7 & 6,7 & \{3,6,7\} \cr\hline
 & \{1,2,3\} &\{4,6,7\} &\{5,6,7\} &
\end{array}
\eee

This implies that Theorem \ref{thm: Two 1-LOCC imply 0-LOCC} is not true without the assumption that $G_A = \overline{G_B}$. As noted above, for this example $G_A = G_B = C_3 \cup C_4$, the union of a triangle with a four-cycle. This has a clique cover number of 5. However, if you add a single diagonal to the four-cycle, you get an intermediate graph which is a subgraph of $\overline{G}_B$ and which has clique cover number 3, and it is this measurement that corresponds to the one-way LOCC measurement. Crucially, this diagonal edge is the same from both Alice's and Bob's perspective  (corresponding in both cases to states $\ket{\psi_6}$ and $\ket{\psi_7}$); it must be part of the initial clique cover, which is why the product measurement does not work.

%Note from David: the 2nd figure we had here for the above example, and the entire 2nd example and Proposition associated with it, I moved to the end of this tex file.

\section{Chordal Graphs and One-Way LOCC}

We next conduct further one-way LOCC analysis on some special cases considered in graph theory.

\begin{defn}
A graph $G$ is chordal if it contains no induced cycles $C_n$ of size greater than three.

A vertex $v$ in a graph $G$ is simplicial if its neighbourhood forms a clique.
\end{defn}

We note that the class of chordal graphs includes trees. Also, in connected chordal graphs, the clique cover number is equal to the minimum vector rank \cite{hackney2009linearly}.

\begin{prop} Suppose that Alice and Bob have a set of mutually orthogonal product states, and let $G_B$ be the graph corresponding to Bob's side. If the graph $G = \overline{G}_B$ is chordal, then the states can be distinguished with one-way LOCC with Alice going first.
\end{prop}

\begin{proof}
The proof of the proposition uses the fact that every chordal graph has a simplicial vertex and that every induced subgraph of a chordal graph is chordal.  This implies the existence of a perfect elimination ordering  of the vertices \cite{rose1970triangulated}, which was used by  Hackney, et al. \cite{hackney2009linearly} to construct an OS-set for chordal graphs. We model our proof on their algorithmic set-up, constructing a direct sum decomposition corresponding to Alice's measurement as in the proof of Theorem~\ref{thm: one-way LOCC graphs}. The algorithm is outlined below. As before, we let $\phi: V_A \rightarrow \H_A$  be the association of vertices with  Alice's states.

\begin{itemize}
\item{} Initialize $j =1$ and put $V_1 = V$.
%\item{}
While $V_{j} \ne \emptyset$, do the following:
\item{} Let $G_{j}$ be the induced subgraph of $G$ on vertices $V_{j}$. $G_{j}$ is an induced subgraph of a chordal graph, so it is chordal.
\item{} Since $G_j$ is chordal, it has a simplicial vertex, which we call $v_j$.
\item{} Define ${\mathcal K}_j = \spn  \{ \phi(v):  v \in V_j, v \not\sim v_j\}$ to be the span of the nonneighbours of $v_j$ in $G_j$.
\item{} If $j= 1$, set ${\mathcal S}_1 =  {\mathcal K}_1^\perp$. Otherwise, set ${\mathcal S}_j =  {\mathcal K}_j^\perp \cap \left( \bigcap_{i = 1}^{j-1} {\mathcal S}_i^\perp\right)$ as in the proof of Theorem \ref{thm: one-way LOCC graphs}.  Since $\phi(v_j) \in  {\mathcal K}_j^\perp$ and $\phi(v_j)$ has a nonzero component in $\left( \oplus_{i = 1}^{j-1} {\mathcal S}_i \right)^\perp$, this space is non-trivial.
\item{} Let $V_{j+1}$ be the set of vertices  $v \in V$ such that $\phi(v)$ has a nonzero component in $\left( \oplus_{i = 1}^j {\mathcal S}_i \right)^\perp$.
\item{} Increase $j$ and iterate.
\end{itemize}
When the process terminates, $V_{j} = \emptyset$ and  $\phi(v) \in \oplus_{i = 1}^j {\mathcal S}_i$ for all $v \in V$. Hence, \bee \oplus_{i = 1}^j {\mathcal S}_i  = \spn\{  \phi(v):  v \in V_j, v \not\sim v_j\} = \H_A . \eee

If two vertices $u$ and $v$  have $\phi(u), \phi(v)$ each with  nonzero components in ${\mathcal S}_i$, then both are neighbours of $v_i$ in $G_i$. Since $v_i$ is simplicial in $G_i$, $u \sim v$ in $G_i$, implying that $u \sim v$ in $G$.

This implies that we have constructed a measurement for Alice that will enable Bob to distinguish his states.
\end{proof}

Consider the following example to help illustrate this proof.

\begin{exa}
{\rm
Let $\phi: V_A \rightarrow \H_A$  be the association of vertices with  Alice's states in $\mathbb{C}^{4}$ given as follows:
\begin{align*}
\phi( v_{1}) &= \ket{0} \\
\phi( v_{2}) &= \ket{1} \\
\phi( v_{3}) &= \ket{0} + \ket{2} + \ket{3} \\
\phi(v_{4})&= \ket{1} + \ket{3}  \\
\phi( v_{5}) &=  \ket{2}.
\end{align*}
Then Alice's graph $G_A$ is given as in Figure \ref{fig:test1}, and suppose that $\overline{G_B}$ is the graph shown in Figure  \ref{fig:test2}. Note that $\overline{G_B}$ is  chordal.

We have labeled the vertices to correspond to the order in which they will be selected by the algorithm.

\begin{figure}[H]
%\centering
%\begin{minipage}{.5\textwidth}
%  \centering
\begin{center}
\definecolor{ududff}{rgb}{0.30196078431372547,0.30196078431372547,1.}
\begin{tikzpicture}[line cap=round,line join=round,>=triangle 45,x=1.0cm,y=1.0cm]
\clip(1.94,0.96) rectangle (7.58,3.72);
\draw [line width=2.pt] (3.46,2.94)-- (4.,2.);
\draw [line width=2.pt] (3.,2.)-- (5.,2.);
\draw (2.56,2.) node[anchor=north west] {$v_{1}$};
\draw (6.32,2.04) node[anchor=north west] {$v_{2}$};
\draw (3.78,1.96) node[anchor=north west] {$v_{3}$};
\draw (4.8,1.98) node[anchor=north west] {$v_{4}$};
\draw (3.22,3.76) node[anchor=north west] {$ v_{5}$};
\draw [line width=2.pt] (5.,2.)-- (6.62,2.);
\draw [line width=2.pt] (3.,-1.)-- (4.54,-1.02);
\draw [line width=2.pt] (4.54,-1.02)-- (6.08,-1.02);
\begin{scriptsize}
\draw [fill=ududff] (3.,2.) circle (2.5pt);
\draw [fill=ududff] (3.46,2.94) circle (2.5pt);
\draw [fill=ududff] (4.,2.) circle (2.5pt);
\draw [fill=ududff] (5.,2.) circle (2.5pt);
\draw [fill=ududff] (6.62,2.) circle (2.5pt);
\draw [fill=ududff] (3.,-1.) circle (2.5pt);
\draw [fill=ududff] (4.54,-1.02) circle (2.5pt);
\draw [fill=ududff] (6.08,-1.02) circle (2.5pt);
\end{scriptsize}
\end{tikzpicture}
\end{center}
%  \captionof{figure}{Alice's Graph, $G_{A}$}
%  \label{fig:test1}
%\end{minipage}%
\caption{Alice's graph, $G_{A}$.}
\label{fig:test1}
\end{figure}
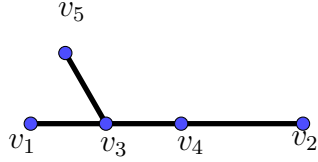

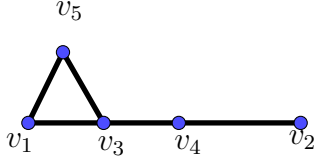
\begin{figure}[H]
\begin{center}
%\begin{minipage}{.5\textwidth}
% \centering
\definecolor{ududff}{rgb}{0.30196078431372547,0.30196078431372547,1.}
\begin{tikzpicture}[line cap=round,line join=round,>=triangle 45,x=1.0cm,y=1.0cm]
\clip(0.84,0.98) rectangle (7.52,4.62);
\draw [line width=2.pt] (3.,2.)-- (3.46,2.94);
\draw [line width=2.pt] (3.46,2.94)-- (4.,2.);
\draw [line width=2.pt] (3.,2.)-- (5.,2.);
\draw (2.56,2.) node[anchor=north west] {$v_{1}$};
\draw (6.32,2.04) node[anchor=north west] {$v_{2}$};
\draw (3.78,1.96) node[anchor=north west] {$v_{3}$};
\draw (4.8,1.98) node[anchor=north west] {$v_{4}$};
\draw (3.22,3.76) node[anchor=north west] {$ v_{5}$};
\draw [line width=2.pt] (5.,2.)-- (6.62,2.);
\draw [line width=2.pt] (3.,-1.)-- (4.54,-1.02);
\draw [line width=2.pt] (4.54,-1.02)-- (6.08,-1.02);
\begin{scriptsize}
\draw [fill=ududff] (3.,2.) circle (2.5pt);
\draw [fill=ududff] (3.46,2.94) circle (2.5pt);
\draw [fill=ududff] (4.,2.) circle (2.5pt);
\draw [fill=ududff] (5.,2.) circle (2.5pt);
\draw [fill=ududff] (6.62,2.) circle (2.5pt);
\draw [fill=ududff] (3.,-1.) circle (2.5pt);
\draw [fill=ududff] (4.54,-1.02) circle (2.5pt);
\draw [fill=ududff] (6.08,-1.02) circle (2.5pt);
\end{scriptsize}
\end{tikzpicture}
\end{center}
  \caption{Complement of Bob's graph, $\overline{G_{B}}$.}
  \label{fig:test2}
\end{figure}

%The simplicial  vertices in $\overline{G_{B}}$ are $ \phi(v_{1}), \phi( v_{2})$ and $\phi( v_{5})$. Using the algorithm.

Now follow the algorithm above:
\begin{itemize}
\item  Iteration $j = 1$: $V_{1} = V \ne \emptyset$.

$G_1 = G  = \overline{G_B}$, and $G_1$ has three simplicial vertices, $v_1, v_2$, and $v_5$. We choose our first vertex $v_1$. The non-neighbours of $v_1$ in $G_1$ are $v_2$ and $v_4$. So $K_{1} = \text{span} \lbrace \phi( v_{2}),\phi( v_{4} )  \rbrace$. Set $S_{1} = K_{1}^{\perp} = \spn{\{\ket{0}, \ket{2}\}}$.
Then $S_{1}^{\perp} =  \spn{\{\ket{1}, \ket{3}\}}$, and $V_2$ is the set of vertices such that $\phi(v_j)$ have a nonzero component in $S_{1}^{\perp}$, so $V_2 = \{ v_2,v_3,v_4 \}$.

\item{} Iteration $j = 2$: $V_2 \ne \emptyset$, so we continue.

$G_2$ is the induced graph on $V_2$, which is simply the path shown in the figures on these vertices. The simplicial vertices are the leaves $v_2$ and $v_3$. We choose our second vertex $v_2$.  The non-neighbour of $v_2$ in $G_2$ is $v_3$,  so $K_{2} = \text{span} \lbrace \phi( v_{3}) \rbrace$. Set  $S_2 = K_{2}^{\perp} \cap S_1^\perp = \spn{\{\ket{1}\}}$, and then observe that
$(S_{1} \oplus S_2)^{\perp} =  \spn{\{\ket{3}\}}$, so $V_3 = \{v_3, v_4\}$.
\item{} Iteration $j = 3$: $V_3 \ne \emptyset$, so we continue.

$G_3$ is simply the edge $\{v_3,v_4\}$. Both vertices are simplicial, and we choose $v_3$. But $G_3$ is a complete graph, so $K_3 = \{0\}$. This means that $S_3 = \C^4 \cap S_1^\perp \cap S_2^\perp = \spn{\{\ket{3}\}}$.

Now we have $(S_1 \oplus S_2 \oplus S_3)^\perp =\{0\}$, so $V_4 = \emptyset$ and the process terminates on the next iteration.
\end{itemize}
This gives the decomposition $\H_A = \C^4 = \spn{\{\ket{0}, \ket{2}\}} \oplus \spn{\{\ket{1}\}} \oplus \spn{\{\ket{3}\}}$. The projections onto these three subspaces give a clique cover of $\overline{G_B}: \{ \{v_1,v_5,v_3\}, \{v_2, v_4\}, \{v_3,v_4\}\}$; and the decomposition of $\H_A$ tells Alice how to make her measurement with respect to this clique cover.
}
\end{exa}

An argument similar to that used in the proposition above can be used to extend the result to graphs that are $k$-trees, as defined below.
We first introduce the  concept of perfect elimination orderings which we will use in our definition of a $k$-tree.

\begin{defn} A perfect elimination ordering of a graph $G$ is an ordering $\{v_1,v_2,...,v_n\}$ of the vertices of $G$ with the property that for all $i$, $v_i$ is a simplicial vertex in the subgraph of $G$ induced by the vertices $\{v_i,v_{i+1},...,v_{n-1},v_n\}$. \end{defn}

It was first shown in \cite{rose1970triangulated} that a graph has a perfect elimination ordering if and only if it is a chordal graph.

We now introduce the class of graphs known as $k$-trees first defined in \cite{beineke1969number} which have been studied extensively, see for example \cite{alinaghipour2014relationship,rose1970triangulated,arnborg1989linear}.

\begin{defn}\label{vardef} A graph $G$ is a $k$-tree if and only if it has a perfect elimination ordering $\{v_1,v_2,...,v_n\}$ with two added properties:

\begin{enumerate}

\item For all $i<n-k$, $v_i$ has degree $k$ in the subgraph of $G$ induced by the vertices $\{v_i,v_{i+1},...,v_{n-1},v_n\}.$

\item The subgraph of $G$ induced by vertices
$
\{v_{n-k},v_{n-k+1},...,v_{n-1},v_n\}
$
is a $k+1$ clique.

\end{enumerate}
$G$ is a partial $k$-tree if it is a subgraph of a $k$-tree.
\end{defn}

It is clear from this definition that every $k$-tree is chordal and that every vertex in a $k$-tree has degree at least $k$.

%\begin{defn}
%For $k \ge 1$, we can define the class of $k$-trees:  $G$ is a $k$-tree if  $G = K_{k+1}$. $G$ is also a $k$-tree if there exists a vertex $v$ such that (a) the neighbours of $v$ form a complete graph $G = K_{k}$ and (b) the graph $G' = G-v$ formed by deleting $v$ is a $k$-tree.
%$G$ is a partial $k$-tree if it is a subgraph of a $k$-tree.
%\end{defn}

%This inductive definition allows us to build up a large class of $k$-trees and partial $k$-trees.

\begin{prop} Suppose that Alice and Bob have a set of mutually orthogonal product states, and let $G_B$ be the graph corresponding to Bob's side. If there exists a graph $G$ such that $G_A \le G \le \overline{G_B}$ and $G$ is a $k$-tree, then the states can be distinguished with one-way LOCC with Alice going first.
\end{prop}

\begin{proof}
We prove this by induction on the number of vertices in $G$. If $G$ has $k+1$ vertices, then $G = K_{k+1}$ and Alice's measurement is irrelevant; Bob can distinguish the states by himself.

Now suppose that the result is true for all $k$-trees with fewer than $n$ vertices and let $G$ be a  $k$-tree on $n$ vertices. Let $v_1$ be the vertex in the perfect elimination ordering from Definition~\ref{vardef}.
Let $\ket{\psi} = \phi(v_1)$, recalling that $\phi$ is an orthogonal representation of $G$, and let Alice measure with $\{ \kb{\psi}{\psi}, I -\kb{\psi}{\psi} \}$. If Alice gets the first outcome, then she knows that the state is in the closed neighbourhood of $v_1$. These are states that can be distinguished by Bob.

If Alice gets the second outcome, her state is now in the state $\phi'(v) = \left(I -\kb{\psi}{\psi} \right)\phi(v)$ for some $v$. Let $G' = G-v$ be the induced subgraph of $G$ formed by deleting $v$.  Let $u$ be any neighbour of $v_1$. Then $u$ has  at least $k$ other neighbours in $G$, which means that there exists a vertex in $V$ that is adjacent to $u$ but not $v_1$. This implies that $\phi(u) \ne \phi(v)$, which implies that $\phi'(u) \ne 0$. Hence $\phi'$ is an orthogonal representation of $G'$.  But $G'$ is a $k$-tree by definition, and so by our inductive assumption, Alice can complete her measurement to put Bob in a position to determine their state.
\end{proof}

%%%%%%%%%%%%%
\section{The Graph of a Domino State Diagram}

The original quantum `domino' states were an orthonormal basis of $\mathbb{C}^3 \otimes \mathbb{C}^3$ consisting entirely of product states.  They were constructed in \cite{bennett1999quantum} as an example of a set of orthogonal product states that cannot be distinguished by LOCC.  A domino diagram was constructed in \cite{bennett1999quantum} to help readers better picture this construction, wherein the orthogonal basis states are represented by dominos and laid on a grid similar to a chessboard (more details are given below).  This has motivated generalizations of domino states in larger Hilbert spaces; examples of these can be found in \cite{cohen2017general,zhang2014nonlocality,zuo2018new}. We will focus on sets of domino states that form a complete product basis of $\mathbb{C}^m \otimes \mathbb{C}^n$; but one can also consider subsets of these bases, such as in Example \ref{Ex: Redundant Orthogonality Qutrits}.

We will define the set of generalized domino states on $\mathbb{C}^m \otimes \mathbb{C}^n$ in terms of the associated domino diagram.  We define a domino diagram to be a partition of the $m \times n$ rectangular chessboard into a set of generalized dominos.  The generalized dominos each have positive integer length and width one and are placed either horizontally or vertically on the chessboard aligned to the $m \times n$ grid. We assume that the chessboard is on a torus, so that dominos which exit off an edge simply continue on the other side.

Label the rows on the chessboard $0,1,2,...,m-1$ from top to bottom and the columns $0,1,2,...,n-1$. This identifies each square on the board with an element of ${\mathbb Z}_m \times {\mathbb Z}_n$. We can then construct a bijection from ${\mathbb Z}_m \times {\mathbb Z}_n$ to a set of generalized domino states in $\C^m \ot \C^n$ as follows:
If there is a horizontal domino on row $r$ whose endpoints are the $(r,b+1)$ and $(r,b+s)$ squares, then the point $(r,b+j)$ with $1\le j\le s$ gets mapped to $\sum_{k=1}^s \alpha_r^k \omega^{jk}\ket{r}\ket{b+k}$ where $\omega$ is the primitive $s$th root of unity and $\alpha_r \ne 0$ is an arbitrary phase associated with row $r$.  If there is a vertical domino on column $c$ whose endpoints are the $(b+1,c)$ and $(b+s,c)$ squares, then the $(b+j,c)$ square with $1\le j\le s$ gets mapped to $\sum_{k=1}^s \beta_c^k \omega^{jk}\ket{b+k}\ket{c}$ where $\omega$ is again the primitive $s$th root of unity and now $\beta_c$ is an arbitrary phase associated with column $c$.

We can obtain the graphs $G_A$ and $G_B$ of a set of generalized domino states directly from its associated domino diagram.

\begin{defn} Let $\mathcal{D}$ be an $m \times n$ domino diagram.  Then the row (respectively column) graph of $\mathcal{D}$ is the graph whose vertex set is the set of $mn$ unit squares on the rectangular chessboard. Two distinct squares are adjacent if and only if one or both of the following two conditions are satisfied:

\begin{enumerate}
\item The two squares lie in the same row (respectively column) on the chessboard.
\item The two squares each lie in two different dominoes $D_1$ and $D_2$ and there exists at least one row (respectively column) of $\mathcal{D}$ which intersects both $D_1$ and $D_2$.
\end{enumerate}  \end{defn}

%\begin{defn} Let $\mathcal{D}$ be an $m$ by $n$ domino diagram.  Then the column graph of $\mathcal{D}$ is the graph whose vertex set is the set of $mn$ unit squares on the rectangular chessboard. Two distinct squares are adjacent if and only if one or both of the following two conditions are satisfied:

%\begin{enumerate}
%\item The two squares lie in the same column on the chessboard.
%\item The two squares each lie in two different dominoes $D_1$ and $D_2$ and there exists at least one column of $\mathcal{D}$ which intersects both $D_1$ and $D_2$.
%\end{enumerate}  \end{defn}

It is an easy exercise to show that the row and column graphs of a domino diagram are the graphs $G_A$ and $G_B$, respectively, of the set of generalized domino states corresponding to the domino diagram.  Since $G_A$ and $G_B$ can never have an edge in common, $G_A\le \overline{G_B}$ and $G_B\le \overline{G_A}$.  We get equality ($G_A= \overline{G_B}$ and $G_B=\overline{G_A}$ ) if and only if any two dominoes in the diagram have either a common row or a common column that intersects them.

%Most of the domino diagrams found in the literature have the following property.

%\begin{defn} We say that an $m$ by $n$ domino diagram is said to be diverse if

%\begin{enumerate}

%\item Any row contains both horizontal and vertical dominoes
%\item Any column contains both horizontal and vertical dominoes

%\end{enumerate}

%A one by one domino is considered to be both a vertical and horizontal domino.

%\end{defn}

The row and column graphs of domino diagrams have a  nice structure that allows us to bound the clique cover number from below:

\begin{prop}\label{dominoprop} Let $G_A$ and $G_B$ be the row and column graphs of an $m \times n$ domino diagram.  Then $\mathrm{cc}(\overline{G_A})\ge n$ and $\mathrm{cc}(\overline{G_B})\ge m$.

If in addition we know that $G_A = \overline{G_B}$, then $\mathrm{cc}(\overline{G_B}) \ge m-v + v^2$ and $\mathrm{cc}(\overline{G_A}) \ge n-h + h^2$, where $v$ and $h$ are the lengths of the longest vertical and horizontal dominoes in the diagram.
\end{prop}

%We can bound the clique covering numbers of the row and column graphs using the following observation: if a domino diagram has two distinct vertical dominoes of length $p$ and $q$ respectively that have a common row that intersects them then the complete bipartite graph $K_{p,q}$ is an induced subgraph of the row graph $G_A$.  (The vertices correspond to the squares of the vertical dominoes).  No two distinct edges in this bipartite graph can lie in the same edge clique.  Hence $cc(G_A)\ge pq$.  The corresponding statement for the column graph $G_B$ and distinct horizontal dominoes also holds by the same logic.

%{\color{red} The original proof was incomplete, since a clique cover can intersect a row in multiple points on a single vertical domino. }

\begin{proof}  Without loss of generality, assume that the largest horizontal domino lies in row 0, and recall the association of domino states with elements of $\mathbb{Z}_m \times \mathbb{Z}_n$. The $n$ states of the form $\{ (0,j) \}$ lie in the first row and thus form a clique in $G_A$ and an independent set in $\overline{G_A}$. Since the clique cover number of a graph is at least its independence number, we have $\mathrm{cc}(\overline{G_A})  \ge n$. This proves the first bound.

For the second bound, let $H$ be the induced subgraph of $\overline{G_A}$ on the $2n$ vertices $\{ (i,j): i \in \{0,1\} \}$. Because the set of states associated with each row form an independent set in  $\overline{G_A}$, the graph $H$ is bipartite. This implies that the clique cover number of $H$ is simply the number of edges in $H$, which is the sum of the degrees of the vertices in a single partition.

If we assume that $G_B = \overline{G_A}$, then each vertex $(0,j)$ is adjacent to $(1,j)$ in $H$, since they are in the same column. This implies that the degree of vertex $(0,j)$ is at least one. If there is a horizontal domino of length $h$ in row zero, then the degree of each of those corresponding vertices is at least $h$, since they are connected to each of the columns they appear in. This implies that if $G_B = \overline{G_A}$,
\bee
\mathrm{cc}(\overline{G_A}) \ge \mathrm{cc}(H) = \sum_{j} \deg(0,j) \ge h(h) + (n-h)(1) = n-h+h^2 .
\eee

The proofs bounding $\mathrm{cc}(\overline{G_B})$ are similar, and the result follows.
\end{proof}

\strut

{\noindent}{\it Acknowledgements.} D.W.K. was partly supported by NSERC and a University Research Chair at Guelph. C.M. was partly supported by Mitacs and the African Institute for Mathematical Sciences. R.P. was partly supported by NSERC. M.N. acknowledges the ongoing support of the Saint Mary's College Office of Faculty Development.

\bibliographystyle{plain}

\bibliography{MNBibfile}

\end{document}